\documentclass[10pt,a4paper]{article}
\usepackage[utf8]{inputenc}
\usepackage[english]{babel}
\usepackage{hyperref}
\usepackage{amsmath,amsthm,enumerate}
\usepackage{amssymb}
\usepackage{thm-restate}
\usepackage[capitalize]{cleveref}
\usepackage{tikz}\usetikzlibrary{calc}\usetikzlibrary{positioning}
\usetikzlibrary{shapes}
\usepackage{subfigure}
\usepackage[hmargin=2.8cm,vmargin=3cm]{geometry}
\usepackage[color=green!30]{todonotes}

\newtheorem{theorem}{Theorem}

\newtheorem{lemma}[theorem]{Lemma}

\newtheorem{definition}[theorem]{Definition}

\newtheorem{claim}{Claim}[theorem]

\newtheorem{remark}[theorem]{Remark}

\newcommand{\ID}{\gamma_{ID}}
\newcommand{\tw}{\text{tw}} 
\newcommand{\pw}{\text{pw}} 
\DeclareMathOperator{\nc}{nc}
\newcommand{\loc}{\gamma_{L}} 

\newcommand{\smallqed}{{\tiny ($\Box$)}\medskip}

\usepackage{tikz}\usetikzlibrary{calc}\usetikzlibrary{positioning}\usetikzlibrary{decorations.pathreplacing}
\usetikzlibrary{decorations.pathmorphing,patterns}

\tikzstyle{noeud}=[circle,inner sep=0, minimum size =7 pt, line width = 1pt, draw=black, fill=white]

\tikzstyle{code}=[circle,inner sep=0, minimum size =6 pt, line width = 1pt, draw=black, fill=black]

\tikzstyle{arete}=[line width = 1pt, black]

\begin{document}

\title{Neighbourhood complexity and identification problems for graphs of bounded treewidth and pathwidth\footnote{Florent Foucaud was financed by the French government IDEX-ISITE initiative 16-IDEX-0001 (CAP 20-25), the International Research Center "Innovation Transportation and Production Systems" of the I-SITE CAP 20-25, and by the ANR project GRALMECO (ANR-21-CE48-0004). Tuomo Lehtil\"a's research was supported  by the Research Council of Finland grants 338797 and 358718.}}
\author{Gaétan Berthe\footnote{\noindent Université Clermont Auvergne, CNRS, Clermont Auvergne INP, Mines Saint-Etienne, LIMOS, 63000 Clermont-Ferrand, France.}
\and Florent Foucaud\footnotemark[2]
\and Tuomo Lehtil\"a\footnote{University of Turku, Department of Mathematics and Statistics, Turku, Finland.}
\and Aline Parreau\footnote{\noindent Univ Lyon, CNRS, INSA Lyon, UCBL, Centrale Lyon, Univ Lyon 2, LIRIS, UMR5205, F-69622 Villeurbanne, France.}
}

\maketitle

\begin{abstract}
The neighbourhood complexity $\nc(G,k)$ of a graph $G$ is a quantity measuring, for a graph $G$ and an integer $k$, the maximum possible number (over all vertex subsets $S$ of size $k$) $|\{N[v]\cap S, v\in V(G)\}|$ of $S$-neighbourhoods in $G$. This notion is important in structural graph theory and algorithm design (especially in parameterized complexity, in particular model checking and kernelization).

While generally $\nc(G,k)\leq 2^k$ and this bound can be achieved, it is known that sparse graphs and structured dense graphs have linear neighbourhood complexity, that is, $\nc(G,k)\in O(k)$ for any such graph $G$. However, for many graph classes, the best possible constants are not known. We focus on graphs of bounded treewidth and pathwidth, showing that (when $k\geq w+1$) (i) if $G$ has treewidth $w\geq 2$, then $\nc(G,k)\leq (k-w+1)2^{w}+w$, and (ii) if $G$ has pathwidth $w\geq 2$, then $\nc(G,k)\leq (k-w+2)2^{w-1}+2k-w-2$. Moreover, we provide constructions that reach these bounds, whenever $w\geq 2$ and $k\geq 2w+1$ ($k\geq 2w-1$ for pathwidth).

Interestingly, in contrast, we also have the tight bound $\nc(G,k)\leq \frac{7}{3}k$, for graphs $G$ with pathwidth~1 or treewidth~1.
\end{abstract}

\section{Introduction}\label{sec:intro}

The neighbourhood complexity of a graph is a notion that, informally speaking, quantifies the maximum possible number of neighbourhoods with respect to a given vertex subset. This notion is important in structural graph theory, as it can for example be used to characterize graph classes of bounded expansion~\cite{Reidl19}.

Typically, graphs of bounded VC-dimension (such as geometric intersection graphs and sparse graphs) have polynomial neighbourhood complexity. While sparse graphs (with bounded expansion, planar, excluding a minor...) and structured dense graphs (of bounded clique-width, twin-width, or merge-width) have linear neighbourhood complexity~\cite{BG25,DBLP:journals/ejc/BonnetFLP24}.

Recently, the neighbourhood complexity has gained increased attraction, as it has proved to be a crucial tool for designing kernelization algorithms~\cite{Eickmeyer17,GAJARSKY}, model checking algorithms~\cite{DBLP:conf/focs/DreierEMMPT24} and other parameterized algorithms~\cite{berthe2024subexponential,DBLP:conf/soda/LokshtanovPSXZ22} for structured graph classes, in connection with the study of structured graph classes and their properties~\cite{BG25,DBLP:journals/corr/abs-2601-14906,chudnovsky2026forbiddinganticompleteplanarminors}. In light of these important applications, it is relevant to obtain good (possibly tight) upper bounds on the neighbourhood complexity for significant graph classes.

Let us formally define neighbourhood complexity (we denote by $N[v]$ the \emph{closed neighbourhood} of $v$, i.e., the set of vertices of $G$ at distance at most~$1$ from $v$).

\begin{definition}[Neighbourhood complexity, $\nc$]
Let $G=(V,E)$ be a graph and let $S$ be a set of vertices of $G$.
The {\em neighbourhood complexity of $S$}, denoted by $\nc(G,S)$, is the number of distinct intersections of closed neighbourhoods of vertices of $G$ within $S$:

$$\nc(G,S)=|\{N[u]\cap S, u\in V\}|.$$

Let $k$ be an integer. We denote by $\nc(G,k)$ the maximum neighbourhood complexity over all vertex subsets of size $k$ of $G$:

$$\nc(G,k)=\max_{S\subseteq V, |S|=k} \nc(G,S).$$
\end{definition}

\paragraph{Connection to identification problems.} In the area of identification problems in graphs (and other combinatorial objects such as hypergraphs or discrete geometric structures), for a given combinatorial object, one wishes to select a (small) solution set, say $S$ (for example, $S$ could be a subset of vertices), so that all elements are uniquely identified by their relation to $S$. For example, a \emph{locating-dominating set} of a graph $G$ is a dominating set $S$ of $G$ (a set of vertices such that every vertex not in $S$ has a neighbour in $S$) such that every vertex in $V(G)\setminus S$ has a unique neighbourhood within $S$~\cite{slater1988dominating}. Such a set is an \emph{identifying code} if in fact, every vertex of $G$ has a unique (closed) neighbourhood within $S$~\cite{karpovsky1998new}. In other words, $S$ is an identifying code if and only if $S$ is a dominating set and $\nc(G,S)=n$, the order of $G$. The smallest size of a locating-dominating set of $G$ is denoted by $\loc(G)$, while that of an identifying code of $G$ is denoted by $\ID(G)$; we will informally refer to these numbers as the \emph{identification numbers} of $G$.

These types of problems have been extensively studied in the literature since the 1960s, under various names: see for example the works of Bondy~\cite{B72} or Rényi~\cite{renyi1961}. We refer to the online bibliography maintained at~\cite{onlineBIBidproblems}, and to the book chapter~\cite{lobstein2023chapter}, for more information and references.

It follows from the definitions that the number of vertices of a graph $G$, with a locating-dominating set $S$, is at most $\nc(G,|S|)+|S|$. If furthermore, $S$ is an identifying code, then $G$ actually has exactly $\nc(G,|S|)$ vertices. 
Hence, upper bounds on $\nc(G,k)$ immediately imply upper bounds on the number of vertices of graphs in terms of their identification numbers (by applying them to the special case where $k$ is the identification number). 
Conversely, the proofs upper-bounding the order of graphs in terms of their identification numbers, such as the ones in~\cite{Bertrand20051,bousquet2015identifying,DBLP:journals/fuin/ChakrabortyFPW24,foucaud2013identifying,FMNPV17,Rall1984location,slater1987domination}, can usually be translated into proofs upper-bounding the neighbourhood complexity of such graphs.

\paragraph{Known bounds for various graph classes.} A notion close to that of neighbourhood complexity has been first considered for hypergraphs under the name of \emph{trace function} or \emph{shatter function}~\cite{AMS19,F83}, defined analogously to neighbourhood complexity, but by replacing neighbourhoods by hyperedges. The \emph{VC-dimension} of a graph $G$ is the size of a largest set $S$ such that $\nc(G,S)=2^{|S|}$~\cite{KKRUW97}, and a similar definition holds for hypergraphs, using the trace function instead of the neighbourhood complexity~\cite{vcdim}.

A classic result called the \emph{Sauer-Shelah lemma} implies that for (hyper)graphs of VC-dimension at most $d$, the trace function (neighbourhood complexity, respectively) is in $O(k^d)$~\cite{S72,Sh72}.  

Typically, geometric graphs have bounded VC-dimension. For example, interval graphs have VC-dimension at most~2 and unit disk graphs have VC-dimension at most~3~\cite{bousquet2015identifying}. 
Planar graphs have VC-dimension at most~4, and $K_{t+1}$-minor-free graphs (which include graphs of treewidth at most~$t$) have VC-dimension at most $t$~\cite{DBLP:journals/dm/BousquetT15,DBLP:journals/dcg/ChepoiEV07}; more generally, graphs of bounded degeneracy have bounded VC-dimension~\cite{DBLP:conf/wea/CoudertCDV24}.
Structured dense graphs, such as graphs of bounded rank-width (equivalently, clique-width)~\cite{DBLP:journals/dm/BousquetT15} or bounded MIM-width also have bounded VC-dimension.\footnote{The latter can be shown using similar arguments as in~\cite[Proposition 3.4]{DBLP:journals/tcs/KangKST17}.} Thus, by the Sauer-Shelah lemma, all graphs in the above-mentioned classes have polynomial neighbourhood complexity. 
 Consult~\cite{bousquet2015identifying} for other graph classes of bounded VC-dimension. 

Furthermore, sparse graphs typically have linear neighbourhood complexity. Indeed, for every graph class $\mathcal C$ with bounded expansion, and for every graph $G$ in $\mathcal C$, we have $\nc(G,k)\leq f(\mathcal C)k$ for some function $f$. 
In fact, $\mathcal C$ has bounded expansion if and only if, for every positive integer $r$ and every graph in $\mathcal C$, its \emph{distance $r$-neighbourhood complexity} is at most $f(\mathcal C, r)k$ for some function $f$~\cite{Reidl19} (the distance $r$-neighbourhood complexity can be seen as the neighbourhood complexity of the $r$th-power of the graph). For the more general graph classes $\mathcal C$ that are nowhere dense, we have for some function $f$ and every $\epsilon>0$, for any $G\in \mathcal C$, $\nc(G,k)\leq f(\mathcal C,\epsilon)k^{1+\epsilon}$~\cite{GAJARSKY}. 

Graph classes of bounded expansion include those excluding a fixed subdivision or a fixed minor, for example planar graphs, and graph classes of bounded treewidth. In the case of a planar graph $G$, one can easily adapt the proofs from~\cite{Rall1984location} for locating-dominating sets to show that $\nc(G,k)\leq 7k-9$, and moreover $\nc(G,k)\leq 4k-2$ if $G$ has treewidth~2, and $\nc(G,k)\leq \frac{7k-1}{2}$ if $G$ is outerplanar. One can also easily deduce from bounds on locating-dominating sets of trees~\cite{slater1987domination} that $\nc(G,k)\leq 3k$ if $G$ is a forest.

Graph classes with bounded merge-width, recently introduced in~\cite{DBLP:conf/stoc/DreierT25}, that form a very nice and general notion of structured dense graphs (or "structurally sparse graphs"), also have linear neighbourhood complexity~\cite{BG25}. Such classes include graph classes of bounded expansion, graph classes of bounded twin-width, clique-width, etc. The case of graphs of bounded twin-width was studied in more detail in~\cite{DBLP:journals/ejc/BonnetFLP24}, where asymptotically tight bounds were provided.

\paragraph{Treewidth and pathwidth.} Treewidth and its refinement pathwidth, in some general sense, measure how close the structure of a graph is to a tree or a path. For formal definitions, we refer to Section~\ref{sec:prelims}. These graph parameters are among the most successful in structural graph theory, notably since the grand Graph Minor project by Robertson and Seymour~\cite{DBLP:journals/jct/RobertsonS04}, in which treewidth and pathwidth play an essential role. They are also very important for graph algorithms, notably due to the immense success of Courcelle's ``algorithmic meta-theorem'' on Monadic Second Order Logic model checking for graphs of bounded treewidth~\cite{DBLP:journals/iandc/Courcelle90}. See Chapter~7 in the book~\cite{Cygan15} for details on several algorithmic applications of treewidth, and~\cite{DBLP:conf/wg/Bodlaender06} for a 2006 survey on the topic. More refined results exist for pathwidth, see for example~\cite{DBLP:conf/icalp/Lampis23}.

Graphs of treewidth at most~$w$ are sparse graphs, as they do not contain large complete graphs as minors and thus have bounded expansion. Hence, it follows that they have linear neighbourhood complexity~\cite{Reidl19}: at most $f(w)k$ for some function $f$. There has already been earlier work focused on the neighbourhood complexity of graphs of given treewidth~\cite{beaudou2025profileneighbourhoodcomplexitygraphs,JR}, however, it considered the distance $r$-neighbourhood complexity of graphs of treewidth~$w$, trying to determine the optimal asymptotic behaviour in terms of $r$ and $w$. The current state of the art is an upper bound of $2^{w+3}(w+1)^{w+1}(r+1)^{w+1}k$~\cite{beaudou2025profileneighbourhoodcomplexitygraphs} which gives for $r=1$ an upper bound of $2^{2w+4}(w+1)^{w+1}k$. A construction witnessing $r$-neighbourhood complexity at least $\Omega(w^{-w}r^{w}k)$ is shown in~\cite{JR}, however note that it requires $r\geq2^w(w+1)$. 
The optimal growth rate in terms of $w$ and $r$ is still open, as noted  in~\cite{beaudou2025profileneighbourhoodcomplexitygraphs,JR}.

When $r=1$, one can deduce a double-exponential upper bound of the form $2^{2^{O(w)}}k$ for the neighbourhood complexity of graphs of treewidth~$w$, by using the known relation that the treewidth is at most exponential in twin-width~\cite{DBLP:conf/wg/JacobP22} and the exponential upper bound for the neighbourhood complexity of graphs of bounded twin-width~\cite{DBLP:journals/ejc/BonnetFLP24}. However, as we will see, this upper bound is far from tight.

\paragraph{Our results.} We determine the best possible bounds on the (distance~1) neighbourhood complexities of graphs of given treewidth and pathwidth.

We start by proving an optimal upper bound on the neighbourhood complexity of graphs of treewidth~1 and pathwidth~1, which turns out to be the same bound. The method extends a proof from~\cite{Bertrand20051} for identifying codes of trees.
\begin{restatable}{theorem}{NCTrees}
\label{thm:NC-trees} 
For any forest $F$ of order $n\geq 3$ and integer $k$, $\nc(F,k)\leq \left\lfloor\frac{7}{3}k\right\rfloor$, and the bound is tight, even for caterpillars\footnote{A \emph{caterpillar} is a tree where all its vertices of degree at least~2 lie on a single path.} of maximum degree~3 (which have pathwidth~1).
\end{restatable}

We then proceed to prove our two main results, for the case $\tw(G)\geq 2$, as follows. In contrast with the case $w=1$, the two bounds are different. We note that we improve the best previously known upper bound by roughly a factor of $2^{w+4}(w+1)^{w+1}$.

\begin{restatable}{theorem}{TWUB}\label{thm:TW-UB}
Let $G$ be a graph of treewidth $w\geq2$ and $k\geq w+1$ an integer. Then, $\nc(G,k)\leq (k-w+1)2^{w}+w$.
\end{restatable}

Interestingly, for pathwidth, we can reduce the dependency on $w$ by essentially a factor of~2, as follows.

\begin{restatable}{theorem}{PWUB}\label{thm:PW-UB}
Let $G$ be a graph of pathwidth $w\geq 2$ and $k\geq w+1$ an integer. Then, $\nc(G,k)\leq (k-w+2)2^{w-1}+2k-w-2$.
\end{restatable}

We also show that both bounds of Theorem~\ref{thm:TW-UB} and Theorem~\ref{thm:PW-UB} are tight for any $w\geq 2$ and $k\geq 2w+1$ ($k\geq 2w-1$ for pathwidth), by providing adequate constructions.

\paragraph{Outline.} We start with some preliminaries in Section~\ref{sec:prelims}. We focus on graphs of treewidth~1 and pathwidth~1 in Section~\ref{sec:trees}. We then show our upper bounds for $w\geq 2$ in Section~\ref{sec:UB}, and the matching lower bounds in Section~\ref{sec:LB}. We state implications for identification numbers in Section~\ref{sec:ID}, and conclude in Section~\ref{sec:conclu}.

\section{Preliminaries}\label{sec:prelims}
\paragraph{General notations}
Let $G$ be a graph and $S\subseteq V(G)$. The subgraph induced by $G$ on $S$ is denoted $G[S]$. We denote by $N[v]$ the \emph{closed neighbourhood} of $v$, i.e., the set of vertices of $G$ at distance at most~$1$ from $v$. For $v\in V(G)$ we say that $N[v]\cap S$ is the $S$-neighbourhood of $v$, the neighbourhood complexity $\nc(G,S)$ is then defined as the the number of distinct $S$-neighbourhoods. Two vertices $x,y$ of $G$ are said to be \emph{separated} by a vertex $s$ of $S$ (or, by extension, separated by $S$) if $s$ is adjacent to exactly one of $x$ and $y$, implying that $x,y$ have distinct $S$-neighbourhoods.

\paragraph{Tree-decomposition and widths}
A \emph{tree-decomposition} of a graph $G$ is a pair $(T,X)$, where $T$ is a tree and $X: V(T) \rightarrow 2^{V(G)}$ is a mapping from $V(T)$ (called \emph{nodes}) to subsets of $V(G)$ (called bags) such that:

1. for all $uv \in E(G)$, there exists a node $t\in V(T)$ such that $\{u,v\} \subseteq X(t)$;
    
2. for all $v \in V(G)$, the subgraph of $T$ induced by $T_v = \{ t \in V(T)~|~v \in X(t)\}$ is a non-empty tree.
\noindent The \emph{width} of a tree-decomposition $(T,X)$ is $\max_{t\in V(T)} |X(t)| - 1$. The \emph{treewidth} of $G$, denoted by $\tw(G)$, is the minimum possible width of a tree-decomposition of $G$.

A tree-decomposition $(T,X)$ of a graph $G$ where the tree $T$ is a path is called a \emph{path-decomposition}. The \emph{pathwidth} of $G$, denoted by $\pw(G)$, is the minimum possible width of a path-decomposition of $G$.

We will need the following lemma about tree-decompositions:

\begin{lemma}\label{lem:noneighbourinfullbag}
Let $G$ be a graph and $(T, X)$ be a tree-decomposition of $G$ of width 
$w$. Let $t$ be a node with bag $X(t)$ of size $w+1$. Then no vertex of $G$ outside of $X(t)$ can be adjacent to the full bag $X(t)$.
\end{lemma}

\begin{proof}
Indeed, suppose to the contrary that there exists a vertex $x$ adjacent in $G$ to the $w+1$ vertices of $X(t)$ and let $x\not\in X(t)$. Then for some two $u,v\in X(t)$ there exist nodes $t'$ and $t''$ with bags containing exactly one of the pairs $\{x,u\}$ and $\{x,v\}$. However, now it is not possible that vertex $x$ simultaneously belongs to bags of connected nodes, vertex $u$ belongs to bags of connected nodes and vertex $v$ belongs to bags of connected nodes, a contradiction.
\end{proof}

\section{Forests}\label{sec:trees}

In this section, we consider the neighbourhood complexity when the underlying graph $G$ has $\tw(G)=1$, namely forests. As Theorem~\ref{thm:TW-UB} requires $\tw(G)=w\geq2$, we have separated this case as its own section, and in fact the best possible upper bound for $w=1$ is a different one.

It is not difficult to show that $\nc(T,k)\in O(k)$ for any tree $T$: for example, using the ideas for locating-dominating sets from~\cite{slater1987domination}, one can easily show that $\nc(T,k)\leq 3k$. For identifying codes, the bound $\frac{7}{3}k+1$ was shown in~\cite{Bertrand20051}, but the bound cannot be immediately translated to neighbourhood complexity. Nevertheless, we next show that one can obtain essentially the same bound for the neighbourhood complexity of trees, and that it is tight (even for trees of pathwidth~1).

\NCTrees*
\begin{proof} Let $S$ be a vertex subset of $F$ of size $k$, let $V_1$ be the set of vertices of $F$ that have exactly one vertex of $S$ in their closed neighbourhood, and $V_{2^+}$ the vertices of $F$ that have at least two vertices of $S$ in their closed neighbourhood. Let $S_1=V_1\cap S$, $S_{2^+}=V_{2^+}\cap S$, and $N_{2^+}=V_{2^+}\setminus S_{2^+}$. 

Let $\mathcal C_S$ be the set of connected components of the subforest $F[S]$ induced by $S$. We create an auxiliary graph $H$ with vertex set  $\mathcal C_S$ as follows. For each vertex $x$ in $N_{2^+}$, we add exactly one edge in the following way. Since $F$ is a forest, $x$ is adjacent to at least two distinct connected components in $F[S]$ (if $x$ had two neighbours in the same connected component of $F[S]$, it would form a cycle, a contradiction). We choose two of these connected components 
and add an edge between the corresponding vertices in $H$. Since $F$ is a forest, the edges of $H$ are all distinct (otherwise we would again get a cycle), hence $H$ has exactly $|N_{2^+}|$ edges. Moreover, $H$ can be obtained from $F$ by deleting vertices/edges and contracting edges, and thus $H$ is a forest. 

Let $c_2$ be the number of connected components in $\mathcal C_S$ that have exactly two vertices. Then, $|\mathcal C_S|\leq |S_1| + c_2 + \frac{|S_{2^+}|-2c_2}{3}$. As $H$ is a forest on $|\mathcal C_S|$ vertices, it has at most $|\mathcal C_S|-1$ edges, and hence, $|N_{2^+}|\leq |\mathcal C_S|-1\leq |S_1| + c_2 + \frac{|S_{2^+}|-2c_2}{3}-1$.

There are at most $|S|$ distinct $S$-neighbourhoods among vertices of $V_1$. There are at most $|S_{2^+}|-c_2$ distinct $S$-neighbourhoods among vertices of $S_{2^+}$, since the two vertices in any connected component in $\mathcal C_S$ of size~2  have the same $S$-neighbourhood. Vertices in $V(F)\setminus (V_1\cup V_{2+})$ account for one (empty) $S$-neighbourhood. Altogether, we obtain:

\begin{align*}
    \nc(F,S) & \leq |S|+|S_{2^+}|-c_2+|N_{2^+}|+1\\
    & \leq |S|+|S_{2^+}|-c_2+|S_1| + c_2 + \frac{|S_{2^+}|-2c_2}{3}-1+1\\
    & \leq k+k+\frac{k}{3}\\
    & = \frac{7k}{3}.
\end{align*}

Finally, by observing that $\nc(F,S)$ is an integer, we then obtain the wanted result $\nc(T,S)\leq \left\lfloor\frac{7k}{3}\right\rfloor$.

To see that the bound is tight, we can consider a construction from~\cite{Bertrand20051}. Let $k=3q+r$ with $0\leq r < 3$, and for any integer $i$, let $T_i$ be the tree on $2i$ vertices obtained by starting with a path on $i$ vertices, and by attaching a distinct leaf to each of these $i$ vertices (see Figure~\ref{fig:tighttreePi}). In a first construction step, consider $(q-1)$ copies of $T_3$, and a copy of $T_{3+r}$. Let $S$ be the set containing the $3$ vertices in the initial paths each of the $(q-1)$ copies of $T_3$, and the $3+r$ vertices of the original path of the $T_{3+r}$ copy. The number of vertices in $S$ is then $3(q-1)+3+r=k$. For the second step of the construction, consider any tree $T'$ on $q$ vertices, with each of these vertices representing one of the $q$ trees obtained in the first step of the construction. Then for every edge $e$ on $T'$, we add to the construction a path on two edges 
between two $S$-vertices of the two copies of $T_3$ or $T_{3+r}$ corresponding to the endpoints of $e$. Finally, attach a single vertex to any vertex of the obtained tree that is not in $S$. See Figure~\ref{fig:tighttreeID-explicit} for an illustration, where $T'$ is a star with three leaves.

It is not difficult to see that every vertex has a distinct $S$-neighbourhood, and the total number of vertices is $2k+(q-1)+1=7q+2r=\lfloor\frac{7k}{3}\rfloor$.

Finally, notice that if $T'$ is a path and the final vertex with empty $S$-neighbourhood is attached to a vertex of degree~2, then the constructed tree is a caterpillar of maximum degree~3.
\end{proof}

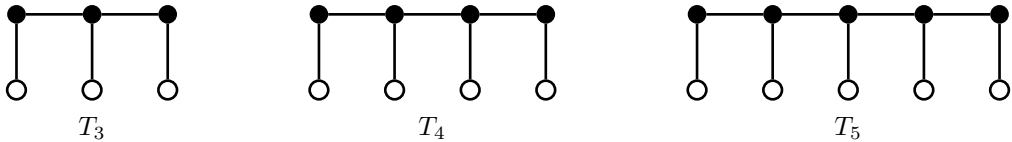
\begin{figure}[!ht]
    \centering
    \begin{tikzpicture}
    
    \node[noeud](n0) at (0,0) {};
    \node[noeud](n1) at (1,0) {};
    \node[noeud](n2) at (2,0) {};
    
    \node[code](c0) at (0,1) {};
    \node[code](c1) at (1,1) {}; 
    \node[code](c2) at (2,1) {};
     
     \foreach \I in {0,1,2}
        \draw[arete] (c\I) -- (n\I);
    
    \draw[arete] (c0)--(c1)--(c2);
    
    \node at (1,-0.5) {$T_3$};
    
    \begin{scope}[shift={(4,0)}]
       \node[noeud](n0) at (0,0) {};
    \node[noeud](n1) at (1,0) {};
    \node[noeud](n2) at (2,0) {};
  \node[noeud](n3) at (3,0) {};
    
    \node[code](c0) at (0,1) {};
    \node[code](c1) at (1,1) {}; 
    \node[code](c2) at (2,1) {};
      \node[code](c3) at (3,1) {};
      
     \foreach \I in {0,1,2,3}
        \draw[arete] (c\I) -- (n\I);
    
    \draw[arete] (c0)--(c1)--(c2)--(c3);
    
    \node at (1.5,-0.5) {$T_4$};
    \end{scope}
    
     \begin{scope}[shift={(9,0)}]
       \node[noeud](n0) at (0,0) {};
    \node[noeud](n1) at (1,0) {};
    \node[noeud](n2) at (2,0) {};
    \node[noeud](n3) at (3,0) {};
       \node[noeud](n4) at (4,0) {};
       
    \node[code](c0) at (0,1) {};
    \node[code](c1) at (1,1) {}; 
    \node[code](c2) at (2,1) {};
      \node[code](c3) at (3,1) {};
       \node[code](c4) at (4,1) {};
       
     \foreach \I in {0,1,2,3,4}
        \draw[arete] (c\I) -- (n\I);
    
    \draw[arete] (c0)--(c1)--(c2)--(c3)--(c4);
    
    \node at (2,-0.5) {$T_5$};
    \end{scope}
    
    \end{tikzpicture}
    \caption{The trees $T_i$ for $3\leq i \leq 5$ used in the proof of Theorem~\ref{thm:NC-trees}. The vertices of $S_i$ are in black.}
    \label{fig:tighttreePi}
\end{figure}

\begin{figure}[!ht]
    \centering
    \begin{tikzpicture}
    
    \node[noeud](n0) at (0,0) {};
    \node[noeud](n1) at (1,0) {};
    \node[noeud](n2) at (2,0) {};
    \node[code](c0) at (0,1) {};
    \node[code](c1) at (1,1) {}; 
    \node[code](c2) at (2,1) {};
    \foreach \I in {0,1,2}
        \draw[arete] (c\I) -- (n\I);
    \draw[arete] (c0)--(c1)--(c2);

    \node[noeud](n3) at (4,0) {};
    \node[noeud](n4) at (5,0) {};
    \node[noeud](n5) at (6,0) {};
    \node[code](c3) at (4,1) {};
    \node[code](c4) at (5,1) {}; 
    \node[code](c5) at (6,1) {};
    \foreach \I in {3,4,5}
        \draw[arete] (c\I) -- (n\I);
    \draw[arete] (c3)--(c4)--(c5);
    
    \node[noeud](n6) at (2,4) {};
    \node[noeud](n7) at (3,4) {};
    \node[noeud](n8) at (4,4) {};
    \node[code](c6) at (2,3) {};
    \node[code](c7) at (3,3) {}; 
    \node[code](c8) at (4,3) {};
    \foreach \I in {6,7,8}
        \draw[arete] (c\I) -- (n\I);
    \draw[arete] (c6)--(c7)--(c8);

    \node[noeud](n9) at (6,4) {};
    \node[noeud](n10) at (7,4) {};
    \node[noeud](n11) at (8,4) {};
    \node[code](c9) at (6,3) {};
    \node[code](c10) at (7,3) {}; 
    \node[code](c11) at (8,3) {};
    \foreach \I in {9,10,11}
        \draw[arete] (c\I) -- (n\I);
    \draw[arete] (c9)--(c10)--(c11);
    
        \node[noeud] (n23) at ($(c2)!0.5!(c3)$) {};
    \node[noeud] (n37) at ($(c3)!0.5!(c7)$) {};
    \node[noeud] (n59) at ($(c5)!0.5!(c9)$) {};

    \node[noeud] (0) at ($(n5)+(1,0)$) {};

    \draw[arete] (c2)--(n23)--(c3)--(n37)--(c7) (c5)--(n59)--(c9) (0)--(n5);

\end{tikzpicture}
    \caption{Example of a tree $T$ with $k=12$ and $\nc(T,k)=28$, as constructed in Theorem~\ref{thm:NC-trees}. The vertices in $S$ are the dark ones.}
    \label{fig:tighttreeID-explicit}
\end{figure}
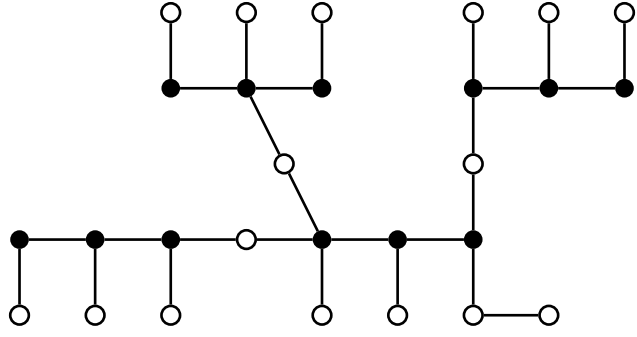

\begin{remark}
We remark that Theorem~\ref{thm:TW-UB} would give $\nc(G,k)\leq 2k+1$ for $w=1$ if it was defined in that case. However, we have the tight upper bound of $\nc(G,k)\leq \lfloor\frac{7}{3}k\rfloor$ for $w=1$ and $2k+1<\lfloor\frac{7}{3}k\rfloor$ when $k\geq 6$. Hence, the neighbourhood complexity behaves differently for the cases of $\tw(G)=1$ and $\tw(G)\geq2$.
\end{remark}

\section{Upper bounds for widths at least 2}\label{sec:UB}

We now proceed with the main results of the paper: the upper bounds for graphs of treewidth and pathwidth at least~2.

\subsection{Treewidth}

In the following, we give the proof for the tight upper bound of neighbourhood complexity with respect to treewidth.

\TWUB*
\begin{proof}
Fix $w\geq 2$ and $k\geq w+1$. Take $G_0$ as a graph with treewidth at most $w$, that maximizes $\nc(G_0,k)$ with a minimum number of vertices, and among the possible graphs take one with a minimum number of edges. By the definition of the neighbourhood complexity, there exists a set $S$ of size $k$ such that the number of neighbourhoods within $S$ is $\nc(G_0,k)$.

As a first observation, for every $u,v\in V(G_0)\setminus S$, by the minimality of the number of vertices we have $N[u]\cap S\neq N[v] \cap S$ and by the minimality of the number of edges, there is no edge between $u$ and $v$.

Secondly, observe that the maximality of $\nc(G_0,k)$ gives the existence of a vertex $v_0\in V(G_0)\setminus S$ such that $N[v]\cap S=\emptyset$, by the minimality of the number of edges we have $N(v_0)=\emptyset$ and by the minimality of the number of vertices, we have that such $v_0$ is unique. We denote $G=G_0-\{v_0\}$, observe that $\nc(G,k)=\nc(G_0,k)-1$ and that the first observation about a pair of vertices outside of $S$ in $G_0$ still holds in $G$.

 Let $s_r$ be a vertex in $S$.
  Let $(T,X)$ be a rooted tree-decomposition of $G$ of width at most $w$. By adding a node and changing the root if necessary, we can assume that the root node $r\in V(T)$ has bag~$\{s_r\}$. 
  For any vertex $s\in S$, let $r_s$ be the node closest to $r$ in $T$ whose bag contains $s$. We call a node $r_s$ an {\em $S$-root node}.
  We can also assume that all $S$-root nodes are distinct. Indeed, if two $S$-root nodes are the same, say, $r_s=r_{s'}$, we create a copy of $r_s$, add it in $T$ between $r_s$ and the parent of $r_s$ ($r_s$ cannot be the root node $r$ of $T$ as $X(r)=\{s_r\}$), and remove $s$ from the bag of the newly created node. Now, $r_s\neq r_{s'}$.
  In the following, we say that such a tree-decomposition, i.e., a rooted tree-decomposition with root bag $\{s_r\}$ and distinct $S$-root nodes, is {\em good}.
 
 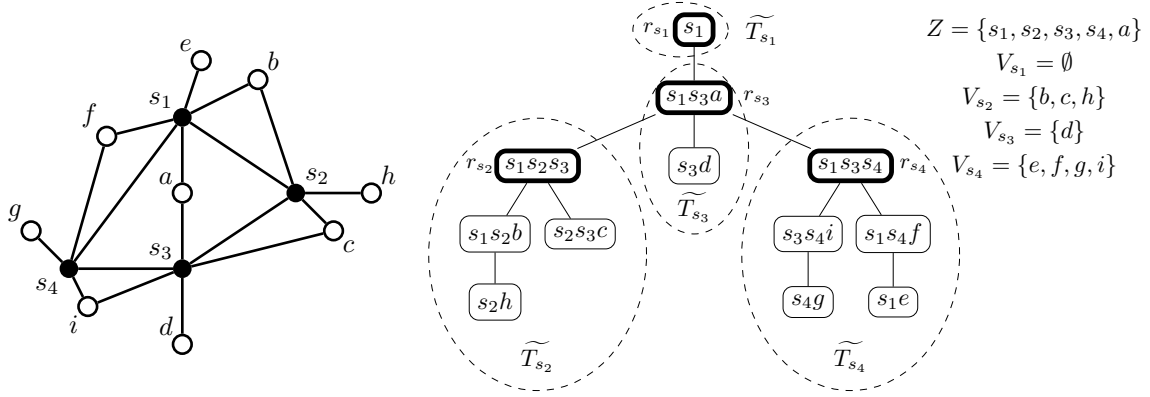
\begin{figure}
     \centering
     \begin{tikzpicture}
     \node[code](1) at (0,0) {};
     \draw (1) node[above left] {$s_1$};
     \node[code](2) at (1.5,-1) {};
         \draw (2) node[above right] {$s_2$};
     \node[code](3) at (0,-2) {};
         \draw (3) node[above left] {$s_3$};
    \node[code](4) at (-1.5,-2) {};
        \draw (4) node[below left] {$s_4$};

    \node[noeud](a) at (0,-1) {};
        \draw (a) node[above left] {$a$};
    \node[noeud](b) at (1,0.5) {};
        \draw (b) node[above right] {$b$};
      \node[noeud](c) at (2,-1.5) {};
          \draw (c) node[below right] {$c$};
       \node[noeud](d) at (0,-3) {};
           \draw (d) node[above left] {$d$};
        \node[noeud](e) at (0.25,0.75) {};
                \draw (e) node[above left] {$e$};
         \node[noeud](f) at (-1,-0.25) {};
             \draw (f) node[above left] {$f$};
        \node[noeud](g) at (-2,-1.5) {};
            \draw (g) node[above left] {$g$};
         \node[noeud](h) at (2.5,-1) {};
             \draw (h) node[above right] {$h$};
              \node[noeud](i) at (-1.25,-2.5) {};
             \draw (i) node[below left] {$i$};
    
    \draw[arete] (e)--(1)--(b)--(2)--(c)--(3)--(2)--(h);
      \draw[arete] (d)--(3)--(a)--(1)--(2) (g)--(4)--(3)--(i)--(4)--(1)--(f)--(4);
    
\node at (0,-3.5) {}; 

     \end{tikzpicture}
     \hfil    
    \scalebox{0.9}{ \begin{tikzpicture}[level distance=1cm,
  level 2/.style={sibling distance=2.3cm},
  level 3/.style={sibling distance=1.25cm}]
  \tikzstyle{every node}=[draw, rounded corners]
 \node[line width=2pt](s1) {$s_1$}
    child {node[line width=2pt](s3) {$s_1s_3a$}
     child {node[line width=2pt](s2) {$s_1s_2s_3$}
        child {node {$s_1s_2b$}
            child {node{$s_2h$}}}
        child {node {$s_2s_3c$}}}
     child {node {$s_3d$}}
     child {node[line width=2pt](s4) {$s_1s_3s_4$}
        child {node {$s_3s_4i$}
            child {node {$s_4g$}}}
        child {node {$s_1s_4f$}
            child {node {$s_1e$}}}}
     }
;
    
     \tikzstyle{every node}=[]
    \draw (s1) node[left=0.2cm] {\small $r_{s_1}$};
  \draw (s2) node[left=0.5cm] {\small $r_{s_2}$};
   \draw (s3) node[right=0.6cm] {\small $r_{s_3}$};
      \draw (s4) node[right=0.6cm] {\small $r_{s_4}$};
 
 \begin{scope}[shift={(5,0.5)}]
     \node at (0,-0.5) {$Z=\{s_1,s_2,s_3,s_4,a\}$};
    \node at (0,-1) {$V_{s_1}=\emptyset$};
\node at (0,-1.5) {$V_{s_2}=\{b,c,h\}$};
\node at (0,-2) {$V_{s_3}=\{d\}$};
\node at (0,-2.5) {$V_{s_4}=\{e,f,g,i\}$};
\end{scope}

\draw[dashed] (-2.3,-3.3) ellipse (1.6cm and 2cm);
\node at (-2.3,-4.8) {$\widetilde{T_{s_2}}$};
\draw[dashed] (2.3,-3.3) ellipse (1.6cm and 2cm);
\node at (2.3,-4.8) {$\widetilde{T_{s_4}}$};

\draw[dashed] (0,-1.75) ellipse (0.8cm and 1.25cm);
\node at (0,-2.6) {$\widetilde{T_{s_3}}$};

\draw[dashed] (-0.2,0) ellipse (0.7cm and 0.4cm);
\node at (1,0) {$\widetilde{T_{s_1}}$};
     \end{tikzpicture}}
     \caption{Notations for the proof of Theorem \ref{thm:TW-UB}, with a graph $G$ of treewidth~2, and set $S=\{s_1,s_2,s_3,s_4\}$, along with a (good) tree-decomposition of $G$.}
     \label{fig:twupperbound}
 \end{figure}

 For the following definitions, an illustrative example is given in Figure~\ref{fig:twupperbound}. For $s\in S$, we define $T_s$ as the subtree of $T$ rooted at $r_s$. We define the subtree $\widetilde{T_s}$ of $T_s$ as the one obtained from $T_s$ by removing from $T_s$ all nodes of each tree $T_{s'}$ with $r_{s'}\in V(T_s)$ and $s\neq s'$. In other words, $\widetilde{T_s}$ contains only the nodes that do not appear in a tree $T_{s'}$ whose root node is strictly farther from $r$ in $T$ than $r_s$ is. 
 Note that if $s'\in X(r_s)$, then $r_{s'}$ is an ancestor of $r_s$ (i.e., it is on the path between $r_s$ and $r$).

  Let $Z\subseteq V(G)$ be the union $\bigcup_{s\in S} X(r_s)$ of the bags of the nodes in $\{r_s~|~s\in S\}$ (in particular $S\subseteq Z$).
  Since $r=r_{s_r}$, every vertex is in a bag of $T_{s_r}$ and thus, every vertex must appear in a bag of a node of some tree $\widetilde{T_s}$. For a given $s\in S$, let us denote by $V_s$ the vertices of $V(G)\setminus Z$ contained in the bags of nodes in $\widetilde{T_s}$. Since every vertex belongs to a bag of a node in some tree $\widetilde{T_s}$, we have $V(G)=Z\cup\bigcup_{s\in S}V_s$. Moreover, since $V_s\cap Z=\emptyset$ for each $s\in S$, the collection of sets $\{V_s~|~s\in S\}\cup \{Z\}$ forms a partition of $V(G)$. Indeed, we cannot have $x\in V_s\cap V_{s'}$ for two distinct $s,s'\in S$ since then $x\in Z$. Moreover, if $x\in V_s$ for some $s\in S$, then $x\not\in S$.
  
  We note that during the rest of the proof, we sometimes modify our good tree-decomposition $(T,X)$ or discuss about different good tree-decompositions. For each such tree-decompositon, we can define $s_r$, $V_s$, $Z$ and other concepts in the same way for these new tree-decompositions as we have done above for $(T,X)$. Then, we could name these objects with a superscript $T$, that is, we could discuss for example about $V_s^T$. However, as adding superscript $T$ to almost every notation could be cumbersome for the reader, we have omitted the superscript $T$ throughout the proof as an abuse of notation.

 \begin{claim}\label{claim:neighboursinX}
For any good tree-decomposition of $G$, for each $s\in S$, for each $x\in V_s$, $N(x)\subseteq X(r_s)$.
 \end{claim}
 
 \begin{claimproof}
  First, since there are no edges between the vertices outside of $S$, all the neighbours of $x$ are in $S$. 
  Assume towards contradiction that there exists a vertex $x\in V_s$ with a neighbour $s'\in S$ which is not in the bag of node $r_s$. Since $xs'$ is an edge, there exists a node $t$ of $T$ whose bag contains both $x$ and $s'$. Since $x\in V_s$, we have $x\notin Z$ and all the bags containing $x$ must belong to nodes of $T_s$ since $T$ is a tree-decomposition and hence, nodes with bags containing $x$ form a connected subtree.
  In particular, $t$ must belong to $T_s$. Consider the path from $t$ to $r$ in $T$: this path contains $r_{s'}$ and $r_s$. Since by our hypothesis, $s'$ is not in the bag of $r_s$, node $r_{s'}$ lies between $t$ and $r_s$ on this path. Since $x$ is not in $Z$, it is not in the bag of $r_{s'}$, and thus,  all the nodes whose bags contain $x$ must belong to $T_{s'}$ and cannot be in $\widetilde{T_{s}}$, a contradiction to the fact that $x\in V_s$. \smallqed
 \end{claimproof}

\medskip

 \begin{claim}\label{claim:V1S}
There exists a good tree-decomposition of $G$ such that $Z=S$.
\end{claim}
  
\begin{claimproof}
Consider a good tree-decomposition $(T,X)$ of $G$ of width at most $w$ that minimizes the number of vertices in $Z\setminus S$. We can assume that for each $s\in S$, each vertex in $V_s$ appears in only one bag that is associated to a leaf node of $T$ adjacent to node $r_s$. Indeed, for each $s\in S$, for each $x\in V_s$, by Claim  \ref{claim:neighboursinX}, we have $N(x)\subseteq X(r_s)$. Thus, we can remove $x$ from all bags, and add a leaf node attached to $r_s$ with bag $N(x)\cup\{x\}$ (since $x$ is not in $X(r_s)$, neighbourhood $N(x)$ has size at most $w$ by Lemma \ref{lem:noneighbourinfullbag}). Note that with this operation we do not add any vertices to $Z$, thus the tree-decomposition is still minimizing the number of vertices in $Z\setminus S$.
Assume towards contradiction that $Z\setminus S$ is nonempty.

We first consider the case where  there exists $x\in Z\setminus S$ such that $N(x)\subseteq X(r_s)$ for some $s$. Note that $x$ has at most $w$ neighbours in $X(r_s)$ by Lemma \ref{lem:noneighbourinfullbag} and thus $N(x)$ has size at most $w$. We transform the decomposition $T$ by removing $x$ from all bags and adding a child node $t$ to $r_s$ with bag $N(x)\cup\{x\}$. This decomposition is still good and has one less vertex in $Z$, contradicting the minimality of $(T,X)$.

Thus, the neighbourhood of any vertex in $Z\setminus S$ is not included in a bag $X(r_s)$ for any $s\in S$.
Consider next a vertex $x$ in $Z\setminus S$ with a smallest degree. 
Let $s_1$ be a vertex in $N(x)$ with an 
$S$-root node farthest from the root. By hypothesis in this second case, $N(x)$ cannot be included in $X(r_{s_1})$, so there exists $s_2\in N(x)$ with $s_2\notin X(r_{s_1})$. We must have $s_1\notin X(r_{s_2})$, otherwise $r_{s_1}$ would be on the path between $r_{s_2}$ and $r$, contradicting the maximal depth of $r_{s_1}$. Note that since $s_2\notin X(r_{s_1})$ and $s_1\notin X(r_{s_2})$, vertices $s_1$ and $s_2$ cannot be together in a bag. In particular, they are non-adjacent.

We next prove that $N(x)=\{s_1,s_2\}$. Assume it is not the case, then there should be a vertex $y$ in $V(G)$ with neighbourhood $\{s_1,s_2\}$, otherwise we could remove all edges between $x$ and $S\setminus \{s_1,s_2\}$, contradicting the minimality of the number of edges of $G$. Note that $y\not\in S$ since then $N[y]\cap S=\{s_1,s_2\}$ would imply that $s_1$ and $s_2$ are adjacent which contradicts previous paragraph.
But then $y$ should not be in $Z$ by minimality of $x$. Thus $y$ should be in some set $V_s$, for $s \in S$. By Claim \ref{claim:neighboursinX}, vertices $s_1$ and $s_2$ should be both in the bag of $X(r_s)$, a contradiction. 
Thus, $N(x)=\{s_1,s_2\}$.

Since $N(x)=\{s_1,s_2\}$, there must be a node $t_1$ whose bag contains $x$ and $s_1$ and a node $t_2$ whose bag contains $x$ and $s_2$. Let $P$ be a path in $T$ between $t_1$ and $t_2$. Then, the bag of each node on $P$ contains $x$. Note that $r_{s_1}$ belongs to this path. Without loss of generality, we can assume that $t_1=r_{s_1}$. 

We transform $T$ by removing $x$ from all bags in $T$, adding $s_2$ in the bags of all nodes in $P$ where it was not (since $x$ was in these bags, they still have size at most $w+1$) and adding a child node $t_3$ to $t_1$ with bag $\{s_1,s_2,x\}$ (recall that $w\geq2$, so this is possible).

If the tree-decomposition is not good, then it means that $r_{s_2}$ has changed and is now the same $S$-root node as another $r_{s'}$ for some $s'\in S$. Let $t$ be this node and $U$ its bag.
We add a new node $t_4$ just above $t$ with bag $U\setminus \{s_2\}$ (if $t$ was actually the root node $r$, then it means that $U=\{s_r,s_2\}$ and $t_4$ is the new root of the tree). Then the $S$-root node of $s_2$ is still $t$ and the $S$-root node of $s'$ is now $t_4$. Thus we obtain a good tree-decomposition but where $x$ is not anymore in $Z$. We did not add any vertices to $Z$ since vertices not in  $Z$ appear only in leaf nodes 
Thus the new tree-decomposition has one less vertex in $Z$, contradicting the minimality of $(T,X)$.
\smallqed
\end{claimproof}

\begin{claim}\label{claim:xadjs}
There exists a good tree-decomposition of $G$ with $Z=S$ 
  such that for each $s\in S$, for each $x\in V_s$, we have $s\in N(x)$.
\end{claim}
  
   \begin{claimproof}
   Assume this is not the case and consider a good tree-decomposition of $G$ that satisfies Claim~\ref{claim:V1S} with a minimum number of vertices not in $S$ that do not satisfy the claim (i.e., vertices $x$ for which there exists $s\in S$ such that $x\in V_s$ but $s\notin N(x)$).
   
  Let $x\in V_s$ be a vertex such that $s\notin N(x)$.
  By Claim~\ref{claim:neighboursinX}, $N(x)\subseteq X(r_s)$. In particular, for each vertex $s'\in N(x)$, the node $r_{s'}$ should be on the path between the root $r$ of $T$ and $r_s$. Moreover, remember that by construction of $G$, $N(x)\neq \emptyset$. Let $s_1$ be a vertex of $N(x)$ with an $S$-root node farthest from the root. Then for each $s'\in N(x)$, node $r_{s_1}$ is on the path between $r_s$ and $r_{s'}$. Since $s'$ appears in the bags of these two nodes, $s'$ should be in $X(r_{s_1})$. This means that $N(x)\subseteq X(r_{s_1})$. 
As before, we transform the decomposition $T$ by removing $x$ from all bags and adding a child node $t$ to $r_{s_1}$ with bag $N(x)\cup\{x\}$. This decomposition is still good, has width at most $w$, satisfies Claim \ref{claim:V1S} (we did not add any vertex to $Z$) and has one less vertex not satisfying the claim, contradicting the minimality of the tree-decomposition. \smallqed
 \end{claimproof}
 
 See Figure \ref{fig:twupperboundfinal} for an example of a tree-decomposition of the graph $G$ of Figure \ref{fig:twupperbound} that satisfies Claim~\ref{claim:xadjs}.
 \quad 

\begin{figure}[h]
    \centering
  
  \scalebox{0.9}{ \begin{tikzpicture}[level distance=1.2cm,
  level 2/.style={sibling distance=1.7cm},
  level 3/.style={sibling distance=1.25cm}]
  \tikzstyle{every node}=[draw, rounded corners]
 \node[line width=2pt](s1) {$s_1$}
    child {node[line width=2pt](s3) {$s_1s_3$}
     child {node[line width=2pt](s2) {$s_1s_2s_3$}
        child {node {$s_1s_2b$}}
        child {node {$s_2h$}}
        child {node {$s_2s_3c$}}}
     child {node {$s_1s_3a$}}
     child {node {$s_3d$}}    
     child {node[line width=2pt](s4) {$s_1s_3s_4$}
        child {node {$s_3s_4i$}}
        child {node {$s_4g$}}
        child {node {$s_1s_4f$}}}
    }
    child {node {$s_1e$}}
;
    
     \tikzstyle{every node}=[]
    \draw (s1) node[left=0.3cm] {$r_{s_1}$};
  \draw (s2) node[left=0.6cm] {$r_{s_2}$};
   \draw (s3) node[left=0.6cm] {$r_{s_3}$};
      \draw (s4) node[right=0.6cm] {$r_{s_4}$};
 
 \begin{scope}[shift={(6.5,0)}]
     \node at (0,-0.5) {$S=\{s_1,s_2,s_3,s_4\}$};
   \node at (0,-1) {$V_{s_1}=\{e\}$};
\node at (0,-1.5) {$V_{s_2}=\{b,c,h\}$};
\node at (0,-2) {$V_{s_3}=\{a,d\}$};
\node at (0,-2.5) {$V_{s_4}=\{f,g,i\}$};
\end{scope}

     \end{tikzpicture}}
     \caption{A tree-decomposition of graph $G$ of Figure \ref{fig:twupperbound} that satisfies Claim \ref{claim:xadjs}.}
     \label{fig:twupperboundfinal}
 \end{figure}
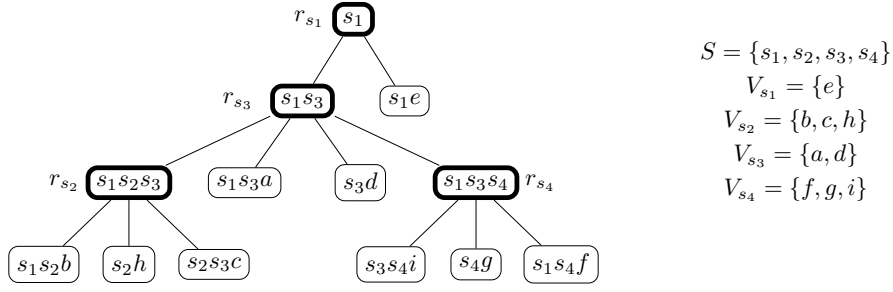

 We can now count the vertices using a good tree-decomposition that satisfies Claim~\ref{claim:xadjs}. Let $n=|V(G)|$. Since $Z=S$, all vertices outside of $S$ are in exactly one set $V_s$. Thus $n-k=\sum_{s\in S}|V_s|$.
 
 Let $w_s=|X(r_s)|$. Since all vertices in $V_s$ have a distinct neighbourhood included in $X(r_s)$ and are adjacent to $s$, there can be at most $2^{w_s-1}$ vertices in $V_s$. Moreover, when $w_s=w+1$, we have $|V_s|\leq 2^w-1$ since a vertex cannot be adjacent to all $w+1$ vertices of $X(r_s)$ without being itself in the bag (see Lemma \ref{lem:noneighbourinfullbag}). 
  We denote $f(w_s)=1$ if $w_s=w+1$ and $f(w_s)=0$ otherwise. Thus,  $|V_s|\leq 2^{w_s-1}-f(w_s)$ and in every case, $|V_s|\leq 2^{w}-1$. 
  
  Since the root vertex $r$ of $T$ has bag $\{s_r\}$, we have $w_{s_r}=1$ and $|V_{s_r}|\leq 2^{1-1}=1$.  Let $r_{s_m}$ be an $S$-root node having a bag of maximum size and let $p=w_{s_m}$ (note that $p\leq w+1$).
  From an $S$-root node $r_{s}$ with $r_s\neq r$ to the next $S$-root node above it, $r_{s'}$, only $s$ can appear in $X(r_s)$ but not in $X(r_{s'})$. As a consequence, in the path from the root $r$ to $r_{s_m}$ there will be, for each $i\in \{1,\ldots,p-1\}$, a vertex $s_i$ with $S$-root bag of size $i$, i.e., such that $|X(r_{s_i})|=w_{s_i}=i$ (note that $s_1=s_r$).
  
  Therefore, for $w\geq2$ and $k\geq w+1$ we have \begin{align*}
      n-k=&\sum_{s\in S}|V_s|\\ 
     \leq&\sum_{i=1}^{p-1}|V_{s_i}| +\sum_{s\in S, s\notin \{s_1,\dots,s_{p-1}\}} |V_s|\\   
      \leq&\sum_{i=1}^{p-1}2^{i-1} +(k-p+1)(2^{p-1}-f(p))\\
      \leq&2^{p-1}-1 +(k-p+1)(2^{p-1}-f(p))\\
      \leq&2^{w}-1 +(k-w)(2^{w}-1)\\
      \leq&(k-w+1)(2^{w}-1).
  \end{align*}
  Above in the second last inequality, we use the assumptions $w\geq2$ and $k\geq w+1$.
  
Finally, by observing that $\nc(G,k)\leq n$ and with the relation $\nc(G_0,k)=\nc(G,k)+1$, we obtain $\nc(G_0,k)\leq (k-w+1)(2^w-1)+k+1=(k-w+1)2^w+w$ as wanted.
 \end{proof}

\subsection{Pathwidth}

We now provide a proof of Theorem~\ref{thm:PW-UB} for pathwidth. It is similar as (and simpler than) the proof of Theorem~\ref{thm:TW-UB}.

\PWUB*
\begin{proof}
As in the proof of Theorem~\ref{thm:TW-UB} for treewidth, we fix the pathwidth $w$ and assume that $k\geq w+1$. Consider $G_0$ as a graph with pathwidth at most $w$ that maximizes $\nc(G_0,k)$ and, subject to this, that has a minimum number of vertices; among all possible graphs, consider one with a minimum number of edges. By the definition of the neighbourhood complexity, there exists a set $S$ of size $k$ such that the number of closed neighbourhoods within $S$ is $\nc(G_0,k)$. By the edge-minimality of $G_0$, there is no edge among the vertices in $V(G_0)\setminus S$. Moreover, for every $u,v\in V(G_0)\setminus S$, by the minimality of the number of vertices, we have $N[u]\cap S\neq N[v] \cap S$. By the maximality of $\nc(G_0,k)$, there is a unique vertex $v_0\in V(G)\setminus S$ such that $N[v]\cap S=\emptyset$. Let $G=G_0-\{v_0\}$, and observe that $\nc(G,k)=\nc(G_0,k)-1$. 

We then consider a path-decomposition $(P,X)$ of width at most $w$, with a natural ordering of the nodes along the path $P$ from one endpoint of $P$ to the other.
Let $p_0$ be the first node of $P$ in this ordering. We can assume that $X(p_0)\cap S$ is nonempty, otherwise there would be no edges among the vertices of $X(p_0)$ and we could remove $p_0$. Let $s_1\in X(p_0)\cap S$. If $X(p_0)$ contains other vertices, then one can add a new node before $p_0$ with bag $\{s_1\}$. Thus, we can assume that $X(p_0)=\{s_1\}$.

For $s\in S$, let $p_s$ be the node closest to $p_0$ whose bag contains $s$ (note that $p_{s_1}=p_0$). We can assume $p_s\neq p_{s'}$ for two vertices $s,s'\in S$. Indeed if $p_s=p_{s'}=p$, then one can add a node $p'$ just before $p$ with bag $X(p')=X(p)\setminus  \{s'\}$. Then $p_s=p'$ and $p_{s'}=p$.

Let $s_1,\dots,s_k$ be the vertices of $S$ in the order they first appear in the bags of $P$. 
 For $i\in\{1,\dots,k-1\}$, we denote by $V^*_{i}$  the vertices of $V(G)\setminus S$  whose last bag containing them is associated to a node between $p_{s_i}$ (included) and $p_{s_{i+1}}$ (excluded). Moreover, $V^*_k$ contains the vertices whose last appearance is in a bag of a node after $p_{s_k}$ ($p_{s_k}$ included).
 Note  that the collection of sets $\{V^*_i\mid i\in \{1,\dots,k\}\}$ forms a partition of $V(G)\setminus S$. 
 Let $V_i$ be the vertices of $V^*_i$ that are not in $X(p_{s_i})$.
 For $x\in V_i$, we have $(N(x)\cap S)\subseteq X(p_{s_i})$ since no vertices of $S$ can appear for the first time in a bag of a node between $p_{s_i}$ and $p_{s_{i+1}}$ (both excluded) or in a node after $p_k$.
 If $|X(p_{s_i})\cap S|\leq w$, let $S_i=X(p_{s_i})\cap S$. Otherwise, all the vertices of $X(p_{s_i})$ are in $S$ and we let $s'\in S$ be a first vertex of $X(p_{s_i})$ that disappears, i.e., that does not appear in the bag of a node located after $p_{s_i}$ (if $p_{s_i}$ is the last node of $P$ we may add a node with a (possibly empty) bag $X(p_{s_i})\setminus \{s'\}$ after it). In this case, let $S_i=X(p_{s_i})\setminus \{s'\}$. In both cases, $|S_i|\leq w$ and $S_i\subset X(p_{s_i})$.
Furthermore, for each $x\in V_i$, $N(x)\subseteq S_i$. Indeed, $x$ appears only in the bags of nodes between $p_{s_i}$ and $p_{s_{i+1}}$ (both excluded), there are no new vertices of $S$ that can appear in the bags of these nodes and, in the case of $|X(p_{s_i})\cap S|=w+1$, vertex $x$ cannot be adjacent to $s'$ since when $x$ appears for the first time, $s'$ has already disappeared. 
In addition, we can assume that $x$ is adjacent to $s_i$, otherwise one could remove $x$ from all bags, and add a node just before $p_{s_i}$ with the bag $(X(p_{s_i})\setminus\{s_i\})\cup\{x\}$.

In conclusion, all vertices of $V_{i}$ have their neighbourhood included in $S_i$ and are adjacent to $s_i$. Hence, $|V_i|\leq 2^{|S_i|-1}$ since they all have distinct $S$-neighbourhoods.

Consider next the vertices in $V^*_i\cap X(p_{s_i})$. Let $n_i=|V^*_i\cap X(p_{s_i})|$ be their number. Recall that the vertices in $V^*_i\cap X(p_{s_i})$ are not in $S$. So as $s_i\in X(p_{s_i})$, we may assume that $n_i\leq w$.
Moreover, since $S_i\subseteq  X(p_{s_i})\cap S$, we have $|S_i|\leq \min(w,w+1-n_i)$.
Thus we have
\begin{align*}
|V^*_i|&=n_i+|V_i| \\
&\leq n_i+2^{|S_i|-1} \\
&\leq n_i+2^{\min(w-1,w-n_i)}.
\end{align*}
Note that this upper bound is maximized for $n_i=1$, with value $2^{w-1}+1$.

We can now do the final counting. Let $n=|V(G)|$. We have $n-k=\sum_{i=1}^k|V^*_i|$.
We count separately the vertices belonging to the first $w+1$ sets $V^*_1,\dots,V^*_{w+1}$ and the vertices belonging to the last set $V^*_k$. First note that all the vertices in $\bigcup_{i=1}^w V^*_i$ can only be adjacent to $s_1,\dots,s_w$ and thus, there are at most $2^w-1$ such vertices.

As a first case, assume $n_{w+1}\geq 1$. Thus, at least one vertex among $s_1,\dots,s_w$, say $s_i$, is not in $X(p_{w+1})$ and thus, cannot be adjacent to $s_{w+1}$. Thus, $s_i$ must also have its neighbourhood included in $\{s_1,\dots,s_w\}$ which means that there can only be $2^w-2$ new distinct neighbourhoods among the vertices of $\bigcup_{i=1}^w V^*_i$  and thus in total  $|\bigcup_{i=1}^{w+1}V^*_i|\leq 2^w-2+2^{w-1}+1=3\cdot2^{w-1}-1$.

For the second case, we have $n_{w+1}=0$. We directly have   $|\bigcup_{i=1}^{w+1}V^*_i|\leq 2^w-1+2^{|S_{w+1}|-1}\leq 3\cdot2^{w-1}-1$.

A similar argument can be used for $V^*_k$. Indeed, if $n_k\geq 1$, then $s_k$ can only be adjacent to vertices in $S_k$. Hence, there will be at most $2^{w-1}-1$ distinct neighbourhoods (and thus vertices) in $V_{k}$ and thus $2^{ w-n_k}-1+n_k\leq 2^{w-1}$ elements in $V^*_k$. If $n_k=0$, then we directly have $|V^*_k|\leq 2^{w-1}$.

Therefore, we have:
\begin{align*}
      n-k=&\sum_{i=1}^k|V^*_i|\\ 
     \leq&\sum_{i=1}^{w+1}|V^*_i| +\sum_{i=w+2}^{k-1} |V^*_i|+|V^*_k|\\   
      \leq& 3\cdot2^{w-1}-1 +(k-w-2)(2^{w-1}+1)+2^{w-1}\\
      \leq & (k-w+2)2^{w-1}+k-w-3\; .
      \end{align*}

  Finally, by observing that $\nc(G,k)\leq n$ and with the relation $\nc(G_0,k)=\nc(G,k)+1$, we obtain $\nc(G_0,k)\leq (k-w+2)2^{w-1}+k-w-3+k+1=(k-w+2)2^{w-1}+2k-w-2$ as desired.
\end{proof}

\section{Tight constructions for widths at least 2}\label{sec:LB}

In this section, we give for each $w\geq2$ and $k\geq 2w\pm1$ the constructions which show that our upper bounds are tight for both tree- and pathwidth. We begin by describing the construction for treewidth.

\subsection{Treewidth}

In this section, for $k\geq 2w+1\geq3$, we give a construction showing that the upper bound of Theorem~\ref{thm:TW-UB} is tight for $w\geq2$. The case with $w=1$ is not tight, but we include it, as it will be used in the proof of Theorem~\ref{thm:PW-tight}. We begin by creating the graph $G_{w,k}$ according to the following rules:
  \begin{enumerate}
      \item Let $S=\{s_1,\ldots,s_k\}$ be a vertex set, with subsets $S_i=\{s_i,\dots,s_{i+w}\}$ for each $1\leq i\leq k-w$.
      \item Create a set $V_1$ of $2^{w+1}-1$ vertices as follows. For each strict  (possibly empty) subset $X$ of $S_1$, create a vertex $v_X\in V_1$ adjacent to each vertex of $X$.
      \item   For each integer $i$ with $2\leq i\leq k-w$, create a set $V_i$ of $2^{w}-1$ vertices as follows.      
      Let $X$ be a distinct strict (possibly empty) subset of $S_i$ which contains $s_{i+w}$. For each such subset $X$, create a vertex $v_X\in V_i$ adjacent to each vertex of $X$.
    \item Create edges so that for any integer $i$ with $1\leq i\leq k-w$, the vertices in $S_i$ form a clique.
  \end{enumerate}

 See Figure \ref{fig:TW-LB} for an illustration of $G_{2,6}$. 
 
  \begin{figure}
    \centering
    \begin{tikzpicture}

    \foreach \I in {-1,0,1,...,14}
         \node[noeud](\I) at (\I,0) {};

\draw[decorate,decoration={brace,mirror,amplitude=10pt}] (-1.2,-0.3) -- node[below=0.4] {$V_1$} (5.2,-0.3);
\draw[decorate,decoration={brace,mirror,amplitude=10pt}] (5.8,-0.3) -- node[below=0.4] {$V_2$} (8.2,-0.3);
\draw[decorate,decoration={brace,mirror,amplitude=10pt}] (8.8,-0.3) -- node[below=0.4] {$V_3$} (11.2,-0.3);
\draw[decorate,decoration={brace,mirror,amplitude=10pt}] (11.8,-0.3) -- node[below=0.4] {$V_4$} (14.2,-0.3);

\foreach \J in {1,2,...,5,6}
         {\node[code](s\J) at (2.3*\J-1,2) {};
         \draw (s\J) ++ (0,0.3) node{$s_\J$} ;}

\draw[arete] (1)--(s1)--(0) (1)--(s2)--(2);
\draw[arete] (s3)--(3)--(s1) (s3)--(4)--(s2) (s3)--(5);

\foreach \I in {6,7,8}
         \draw[arete] (s4)--(\I);

\draw[arete] (6)--(s2) (7)--(s3);

\foreach \I in {9,10,11}
         \draw[arete] (s5)--(\I);

\draw[arete] (9)--(s3) (10)--(s4);

\foreach \I in {12,13,14}
         \draw[arete] (s6)--(\I);

\draw[arete] (12)--(s4) (13)--(s5);

\draw[arete] (s1)--(s2)--(s3)--(s4)--(s5)--(s6);

\path[draw,arete] (s1) to[out=30,in=150] (s3);
\path[draw,arete] (s2) to[out=30,in=150] (s4);
\path[draw,arete] (s3) to[out=30,in=150] (s5);
\path[draw,arete] (s4) to[out=30,in=150] (s6);

    \end{tikzpicture}
    \caption{Illustration of the construction of $G_{w,k}$ for treewidth $w=2$ and $k=6$. All the neighbourhoods to the black vertices are unique.}
    \label{fig:TW-LB}
\end{figure}
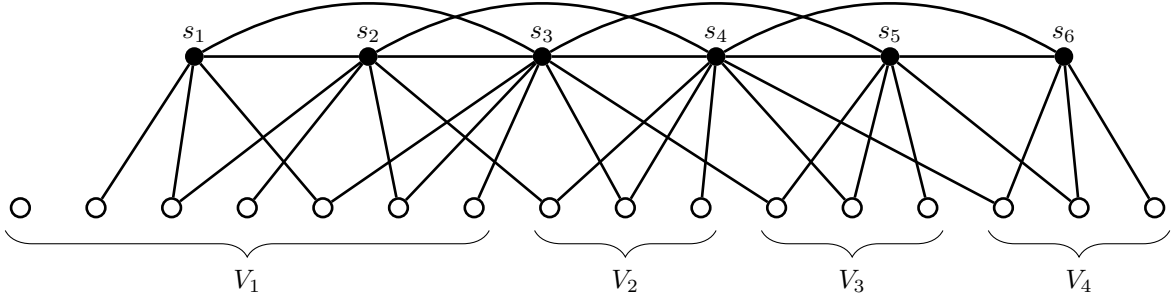

\begin{theorem}\label{thmTreewidth construction}
Let $k\geq 2w+1\geq 3$. The graph $G_{w,k}$ has treewidth $w$, neighbourhood complexity and order $\nc(G_{w,k},k)=\nc(G_{w,k},S)=n=(k-w+1)2^w+w$.
\end{theorem}
\begin{proof}

Let us first show that the graph $G_{w,k}$ has treewidth $w$. We have $\tw(G_{w,k})\geq w$ as the complete graph $K_{w+1}$ is a subgraph of $G_{w,k}$. For bounding the treewidth from above, we first create, for each set $S_i$, a node with bag $S_i$ and an edge between the node with bag $S_i$ to the node with bag $S_{i+1}$, for each $1\leq i\leq k-w-1$. These nodes form a path. Next, for each vertex $v_X\in V_i$ with $N(v_X)=X\subset S_i$, we create a node with bag $X\cup\{v_X\}$ together with an edge to the node with bag $S_i$. Since each such set $X$ is a strict subset of $S_i$, we have $|X\cup\{v_X\}|\leq w+1$. Observe that the created graph is a tree (a path to which multiple leaves are attached, that is, a caterpillar), each bag has at most $w+1$ vertices, and every edge of $G_{w,k}$ belongs to at least one of the bags. Moreover, every vertex is present in the bags of a connected subtree of this tree. Hence, this defines a tree-decomposition of $G_{w,k}$ of width~$w$, and $\tw(G_{w,k})=w$.

Let us next show that the order of graph $G_{w,k}$ is $n=(k-w+1)2^w+w$. 
We have $\sum_{i=0}^{w}\binom{w+1}{i}$ strict subsets $X$ of $S_1$ and thus, $|V_1|=2^{w+1}-1$. Similarly, we have $\sum_{i=0}^{w-1}\binom{w}{i}$  strict subsets $X$ of $S_i$ for $2\leq i\leq k-w$ and hence, $|V_i|=2^w-1$. Therefore, we have \begin{align*}
    n=&|S| + |V_1| + \sum_{i=2}^{k-w} |V_i|\\ 
    =&k + \sum_{i=0}^{w}\binom{w+1}{i} + (k-w-1)\sum_{i=0}^{w-1}\binom{w}{i}\\
    =& k + (2^{w+1}-1) + (k-w-1)(2^w-1)\\
    =& (k-w+1)2^w+k-1-(k-w-1)\\
    =&(k-w+1)2^w+w.
\end{align*}

By definition, $\nc(G_{w,k},k)\leq n$. Thus, it is enough to show $\nc(G_{w,k},k)\geq \nc(G_{w,k},S)=n$. In other words, we need to show that every vertex has a unique neighbourhood in $S$. 
Indeed, by definition no two vertices outside of $S$ have identical neighbourhoods in $S$. Furthermore, the same is true for each vertex in the set $S$ since $k\geq2w+1$ (to pairwise separate vertices in $S$) and since sets $X$ are strict subsets of $S_i$ for some $i$, the vertices in $S$ are separated from the vertices outside of $S$.
\end{proof}

\subsection{Pathwidth}\label{secPathwidthConst}

In this section, we provide a construction showing that the upper bound obtained in Theorem~\ref{thm:PW-UB} is tight for every $k\geq 2w-1\geq3$.

\begin{theorem}\label{thm:PW-tight}
For any fixed integers $w\geq2$ and $k\geq 2w-1$, there is a graph $G$ of pathwidth $w$, neighbourhood complexity and order $\nc(G,k)=n=(k-w+2)2^{w-1}+2k-w-2$.
\end{theorem}

\begin{proof}
  Let $k\geq2w-1\geq 3$.
  We construct the graph $G$ as follows. 
  \begin{enumerate}
      \item Let $S=\{s_1,\ldots,s_k\}$ be a vertex set with subsets $S_i=\{s_i,\dots,s_{i+w-1}\}$ for each $1\leq i\leq k-w+1$.
      \item Create a set $V_1$ of $2^w$ vertices as follows. For each subset $X$ of $S_1$, create a vertex $v_X\in V_1$ adjacent to each vertex of $X$.
      \item   For each integer $i$ with $2\leq i\leq k-w+1$, create a set $V_i$ of $2^{w-1}$ vertices as follows.      
      For each distinct subset $X$ of vertices of $S_i$ which contains $s_{i+w-1}$, create a vertex $v_X\in V_i$ adjacent to each vertex in $X$.

            \item Remove from $V_1$ the vertex adjacent to each vertex in $S_1$ and from $V_{k-w+1}$ the vertex adjacent to each vertex in $S_{k-w+1}$.
      \item Create a set of vertices $Z=\{z_1,\dots,z_{k-w}\}$ so that $z_i$ is adjacent to $s_{i}$ and $s_{w+i}$.
      \item Add edges so that each set $S_i$ forms a clique. 
 \end{enumerate}
  
 See Figure~\ref{fig:PW-LB} for an illustration of the construction for $w=3$ and $k=6$.
  
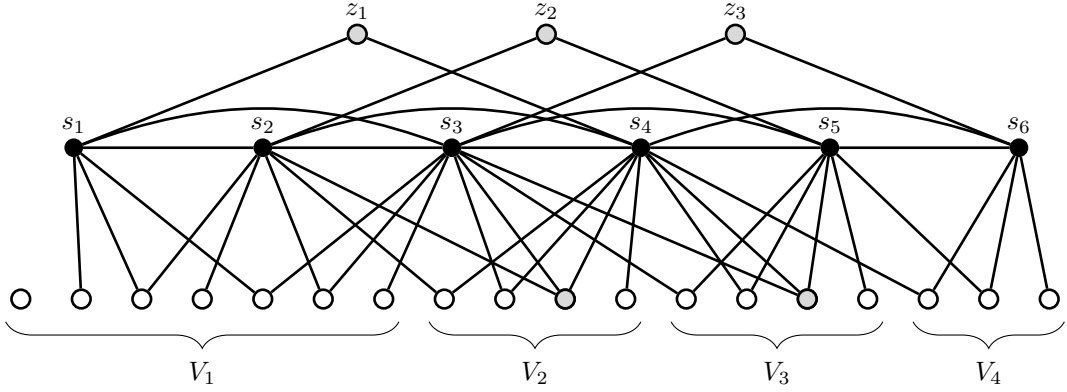
\begin{figure}
    \centering
    \begin{tikzpicture}

    \foreach \I in {1,...,18}
         \node[noeud](\I) at (0.8*\I,0) {};

\foreach \I in {10,14}
\draw (\I) node[noeud, fill=gray!30] {};

\foreach \J in {1,2,3}
{\node[noeud, fill=gray!30](z\J) at (2.5*\J+2.25+0.5,3.5) {};
         \draw (z\J) ++ (0,0.3) node{$z_\J$} ;}

         \draw[decorate,decoration={brace,mirror,amplitude=10pt}] (0.6,-0.3) -- node[below=0.4] {$V_1$} (5.8,-0.3);
\draw[decorate,decoration={brace,mirror,amplitude=10pt}] (6.2,-0.3) -- node[below=0.4] {$V_2$} (9,-0.3);
\draw[decorate,decoration={brace,mirror,amplitude=10pt}] (9.4,-0.3) -- node[below=0.4] {$V_3$} (12.2,-0.3);
\draw[decorate,decoration={brace,mirror,amplitude=10pt}] (12.6,-0.3) -- node[below=0.4] {$V_4$} (14.6,-0.3);

\foreach \J in {1,2,...,5,6}
         {\node[code](s\J) at (2.5*\J-1,2) {};
         \draw (s\J) ++ (0,0.3) node{$s_\J$} ;}

\draw[arete] (2)--(s1) (s2)--(3)--(s1) (4)--(s2);
\draw[arete] (s3)--(5)--(s1)  (s3)--(6)--(s2) (7)--(s3);

\foreach \I in {8,...,11}
         \draw[arete] (s4)--(\I);

\draw[arete] (8)--(s2) (9)--(s3) (s2)--(10)--(s3);

\foreach \I in {12,...,15}
         \draw[arete] (s5)--(\I);

\draw[arete] (12)--(s3)--(14) (13)--(s4)--(14);

\foreach \I in {16,...,18}
         \draw[arete] (s6)--(\I);

         \draw[arete]  (16)--(s4);

         \draw[arete] (s5)--(17);

\draw[arete] (s1)--(s2)--(s3)--(s4)--(s5)--(s6);

\path[draw,arete] (s1) to[out=20,in=160] (s3);
\path[draw,arete] (s2) to[out=20,in=160] (s4);
\path[draw,arete] (s3) to[out=20,in=160] (s5);
\path[draw,arete] (s4) to[out=20,in=160] (s6);

\draw[arete] (s1)--(z1)--(s4);
\draw[arete] (s2)--(z2)--(s5);
\draw[arete] (s3)--(z3)--(s6);

    \end{tikzpicture}
    \caption{Illustration of the construction of $G$ in the proof of Theorem~\ref{thm:PW-tight} for pathwidth $w=3$ and $k=6$. All the neighbourhoods to the black vertices are distinct. Removing the five gray vertices, we obtain the graph $G_{2,6}$ of Figure \ref{fig:TW-LB}.}
    \label{fig:PW-LB}
\end{figure}

  First, we may observe that the total number of vertices $n$ is \begin{align*}
      n=&|S|+|V_1|+\sum_{i=2}^{k-w+1}|V_i|+|Z|\\
      =&k+(2^{w}-1)+\left(\sum_{j=w+1}^{k}2^{w-1}-1\right)+(k-w)\\
      =&(k-w+2)2^{w-1}+2k-w-2.
  \end{align*}
  
   Hence, the claim follows by first showing that $\nc(G,S)=n$ and then that the pathwidth of $G$ is $w$.\smallskip
  
  First, let us show that $\nc(G,S)=n$. Denote by $V_S=\{v\in V(G)\setminus S \mid N(v)=S_i, 1\leq i\leq k-w+1\}$. Observe that $G-Z-V_S$   is an induced subgraph of $G_{w-1,k}$. By Theorem~\ref{thmTreewidth construction}, we have $\nc(G_{w-1,k},S)= |V(G_{w-1,k})|$. Therefore, each vertex in $G-Z-V_S$ has a unique neighbourhood within $S$. Moreover, the vertices of $V_S$ are adjacent to distinct sets of exactly $w$ vertices of $S$. Thus, they are separated from other vertices in $V(G)\setminus (Z\cup S)$. Since we have removed the vertices with neighbourhoods $S_1$ and $S_{k-w+1}$ from $G$ at the fourth step of the construction, also the vertices of $S$ are separated from the vertices of $V_S$. Hence, each vertex in $G-Z$ has a unique neighbourhood in $S$. 
  
  Let us next consider the vertices in $Z$. Let $z_i,z_j\in Z$ with $i<j$. 
  Vertex $s_i$ separates $z_i$ and $z_j$. Moreover, $z_i$ is the only vertex outside of $S$ adjacent to $s_i$ and $s_{i+w}$. Finally, $z_i$ is separated from $s_i$ by $s_{i-1}$ or $s_{i+1}$, and from $s_{i+w}$ by $s_{i+w-1}$ or $s_{i+w+1}$. As $z_i$ is not adjacent to other vertices of $S$, vertex $z_i$ is separated from each vertex of $S$.

  Let us then show that $G$ has pathwidth $w$. We construct the bags of the path-decomposition in the following sequence, where $v_i\in V_i$ for each $i$ (for simplicity we present only one of those bags for each $V_i$):
  \begin{align*}
       &\{s_1\} ,   \{s_1,v_1\} ,  \{s_1,s_2\} ,  \{s_1,s_2,v_2\} ,  \{s_1,s_2,s_3\},\dots,
       \{s_1,s_2,\dots,s_w,z_{w+1}\}, \{s_2,\dots,s_{w+1},z_{w+1}\},\\
       &\{s_2,\dots,s_{w+1},v_{w+1}\} ,  \{s_2,\dots,s_{w+1},z_{w+2}\},   \{s_3,\dots,s_{w+2},z_{w+2}\} ,     \{s_3,\dots,s_{w+2},v_{w+2}\} ,   \\
       &\{s_3,\dots,s_{w+2},z_{w+3}\},\dots, \{s_{k-w},\dots,s_{k-1},z_{k}\} ,  \{s_{k-w+1},\dots,s_{k},z_{k}\} , \{s_{k-w+1},\dots,s_{k},v_k\}.
  \end{align*}

  We may notice that in this sequence, there are at most $w+1$ vertices in any bag. Furthermore, bags containing the same vertex form a connected subpath. Finally, each edge is included in one of the bags. Thus, this defines a path-decomposition of $G$, which has pathwidth at most $w$. Furthermore, by Theorem~\ref{thm:PW-UB}, it is at least $w$ (otherwise, the neighbourhood complexity would be strictly smaller). Hence, the claim follows.   
  \end{proof}

\section{Application to identification problems}\label{sec:ID}

We now state variations of our results for identification problems. For forests, it is known that a forest with a locating-dominating set of size $k$ has order at most $3k-1$ and the bound is tight~\cite{slater1987domination}, and if it has an identifying code of size $k$, then it has order of at most $\lfloor\frac{7k}{3}-1\rfloor$, and this bound is also tight, even for caterpillars (that have pathwidth~1)~\cite{Bertrand20051}.

For graphs of treewidth and pathwidth at least~2, almost identical proofs as in Section~\ref{sec:UB} can be applied to the setting of identification problems, leading to the following theorems.

\begin{theorem}\label{thm:TW-UB-ID}
Let $G$ be a graph of order $n$, treewidth $w\geq2$ with a locating-dominating set (or identifying code) of size $k\geq w+1$. Then, $n\leq (k-w+1)2^{w}+w-1$. Hence, $\ID(G)\geq\loc(G)\geq \frac{n-w+1}{2^{w}}+w-1$.
\end{theorem}
\begin{proof}
The proof is the same as that of Theorem~\ref{thm:TW-UB}, except that we consider $S$ as a locating-dominating set. Hence, all vertices in $V(G)\setminus S$ have a neighbour in $S$, leading to a bound smaller by~1, and all vertices have by definition a distinct neighbourhood within $S$.

We then fix $k$ and $n$ to be the maximum order of a graph with treewidth $w$ and a locating-dominating set of size $k$. Among all such graphs of order $n$ with a locating-dominating set $S$ of size $k$, we let $G$ be one that has the smallest possible number of edges. The rest of the proof is exactly the same.
\end{proof}

Similarly, we have the following lower bound for graphs of given pathwidth. In this case, the bounds are slightly different for locating-dominating sets and identifying codes.

\begin{theorem}\label{thm:PW-UB-ID}
Let $G$ be a graph of order $n$, pathwidth $w\geq2$ with a locating-dominating set of size $k\geq w+1$. Then, $n\leq (k-w+2)2^{w-1}+2k-w-1$, and $\loc(G)\geq \frac{n+w+1-(w-2)2^{w-1}}{2^{w-1}+2}$. If the set is an identifying code, then $n\leq (k-w+2)2^{w-1}+2k-w-3$, and $\ID(G)\geq \frac{n+w+3-(w-2)2^{w-1}}{2^{w-1}+2}$.
\end{theorem}

\begin{proof}
The proof is almost the same as that of Theorem~\ref{thm:PW-UB}, except that we consider $S$ as a locating-dominating set and that we count vertices and not distinct neighbourhoods. This slightly changes the final counting since vertices in $S$ might have the same neighbourhood as a vertex not in $S$. 

We fix $k$ and $n$ to be the maximum order of a graph with pathwidth $w$ and a locating-dominating set of size $k$. Among all such graphs of order $n$ with a locating-dominating set $S$ of size $k$, we let $G$ be one that has the smallest possible number of edges. The rest of the proof is the same until the final counting, where we do not have to consider separately the computation of $|V^*_{w+1}|$ and $|V^*_k|$, which leads to: 

\begin{align*}
      n-k=&\sum_{i=1}^k|V^*_i|\\ 
     \leq&\sum_{i=1}^{w}|V^*_i| +\sum_{i=w+1}^{k} |V^*_i|\\   
      \leq& 2^{w}-1 +(k-w)(2^{w-1}+1)\\
      \leq & (k-w+2)2^{w-1}+k-w-1\; .
      \end{align*}
  
  Note that all vertices in $V(G)\setminus S$ have a neighbour in $S$. Thus we do not need to add 1 for the empty neighbourhood and we obtain that $n\leq (k-w+2)2^{w-1}+2k-w-1$ as desired.
 
 \medskip
 
 For identifying codes, the proof is exactly the same as for the neighbourhood complexity, except that we do not count the empty neighbourhood, and thus we obtain the bound $n\leq (k-w+2)2^{w-1}+2k-w-3$.
\end{proof}

Moreover, note that our constructions from Section~\ref{sec:LB} actually provide identifying codes of size $k$ (with the small change that we should remove the unique vertex not dominated by $S$ in the constructions), and hence these bounds are tight. For locating-dominating sets and pathwidth, one should keep the two vertices that are removed at the fourth step of the construction, to obtain a tight bound.

\section{Conclusion}\label{sec:conclu}

We have determined tight upper bounds for the neighbourhood complexity of graphs of given treewidth or pathwidth and provided constructions attaining them. As further work, one could similarly determine optimal bounds for other classes of sparse graphs or dense structured graphs.

For the more general notion of distance $r$-neighbourhood complexity, bounds in $O(f(w)r^{w+1}k)$ and $\Omega(r^{w}k)$ (for $r\geq 2^w(w+1)$) are known for graphs of treewidth~$w$, but the exact growth rate in terms of $w$ and $r$ is open~\cite{beaudou2025profileneighbourhoodcomplexitygraphs,JR}. The $r$-neighbourhood complexity is also not fully understood for other types of graphs, such as, graphs of bounded twin-width~\cite{DBLP:journals/ejc/BonnetFLP24} or planar graphs~\cite{JR}.

\paragraph*{Acknowledgements.} Florent Foucaud and Aline Parreau thank Fionn McInerney for preliminary discussions held in 2019 on the topic of this paper.

\bibliographystyle{abbrv}
\bibliography{references}

\end{document}